\documentclass[a4paper]{article}
\usepackage{fullpage}
\usepackage{amsmath,amssymb,amsthm}
\usepackage{algorithm,algpseudocode}
\usepackage{tikz}
\usetikzlibrary{snakes}
\usepackage{xspace}
\usepackage{enumerate}
\usepackage{color}
\usepackage{url}
\usepackage{verbatim}
\usepackage{multirow}
\usepackage{authblk}
\usepackage{hyperref}

\newtheorem{fact}{Fact}[section]
\newtheorem{theorem}{Theorem}[section]

\newtheorem{lemma}[theorem]{Lemma}
\newtheorem{definition}[theorem]{Definition}
\newtheorem{observation}[theorem]{Observation}
\newtheorem{proposition}[theorem]{Proposition}
\newtheorem{example}[theorem]{Example}
\newtheorem{remark}[theorem]{Remark}

\newcommand{\floor}[1]{\left\lfloor #1 \right\rfloor}
\newcommand{\ceil}[1]{\left\lceil #1 \right\rceil}
\newcommand{\Oh}{\mathcal{O}}
\newcommand{\poly}{\mathrm{poly}}

\newcommand{\prob}[1]{\mathrm{Pr} [#1]}

\newcommand{\eps}{\varepsilon}
\newcommand{\heavy}[1]{\mathcal{H} \left(#1\right)}

\newcommand{\MwDR}{\textsc{Mismatch with Error Correcting}\xspace}
\newcommand{\Mis}{\textsc{Mismatch}\xspace}
\newcommand{\EPM}{\textsc{Pattern Matching}\xspace}
\newcommand{\WPM}{\textsc{Weighted Pattern Matching}\xspace}
\newcommand{\SWP}{\textsc{Sliding Window Product}\xspace}
\newcommand{\DM}{\textsc{Multiple Pattern Matching}\xspace}

\newcommand{\T}{\mathcal{T}}
\newcommand{\PP}{\mathcal{P}}
\renewcommand{\S}{\mathcal{S}}

\begin{document}
  \title{
    Streaming $k$-mismatch with error correcting and applications
  }

  \author[1]{Jakub Radoszewski}
  \author[2]{Tatiana Starikovskaya}

\affil[1]{
	Institute of Informatics, University of Warsaw, Warsaw, Poland\\
	\texttt{jrad@mimuw.edu.pl}}

\affil[2]{
  DIENS, \'{E}cole normale sup\'{e}rieure, PSL Research University, Paris, 75005, France\\
  \texttt{tat.starikovskaya@gmail.com}
}

\date{\vspace{-.8cm}}

  \maketitle              

\begin{abstract}
We present a new streaming algorithm for the $k$-\textsc{Mismatch} problem, one of the most basic problems in pattern matching. Given a pattern and a text, the task is to find all substrings of the text that are at the Hamming distance at most $k$ from the pattern. Our algorithm is enhanced with an important new feature called \textsc{Error Correcting}, and its complexities for $k=1$ and for a general $k$ are comparable to those of the solutions for the $k$-\textsc{Mismatch} problem by Porat and Porat (FOCS 2009) and Clifford et al.\ (SODA 2016). In parallel to our research, a yet more efficient algorithm for the $k$-\textsc{Mismatch} problem with the \textsc{Error Correcting} feature was developed by Clifford et al.\ (SODA 2019). Using the new feature and recent work on  streaming \textsc{Multiple Pattern Matching} we develop a series of streaming algorithms for pattern matching on weighted strings, which are a commonly used representation of uncertain sequences in molecular biology. We also show that these algorithms are space-optimal up to polylog factors.

A preliminary version of this work was published at DCC 2017 conference~\cite{DBLP:conf/dcc/RadoszewskiS17}.
\end{abstract}

\section{Introduction}
\label{sec:intro}
	In this work we design efficient streaming algorithms for a number of problems of approximate pattern matching. In this class of problems we are given a pattern and a text and wish to find all substrings of the text that are ``similar'' to the pattern. We assume that the text arrives as a stream, one symbol at a time. 

The first small-space streaming algorithms for pattern matching were suggested in the pioneering paper by Porat and Porat in FOCS 2009~\cite{Porat:09}. In particular, they showed an algorithm for the \EPM problem, where one must find all exact occurrences of the pattern in the text. For a pattern of length $m$ their algorithm takes only $\Oh(\log m)$ space and $\Oh(\log m)$ time per each symbol of the text, and reports all exact occurrences of the pattern in the text as they occur. In 2010, Erg{\"{u}}n et al.\ \cite{PeriodicityInStreams} presented a slightly simpler version of the algorithm, and finally in 2011 the running time of the algorithm was improved to constant by Breslauer and Galil~\cite{DBLP:journals/talg/BreslauerG14}. All of the aforementioned results are randomised Monte Carlo algorithms.

The next logical step was to study the complexity of approximate pattern matching in the streaming model in which we are to find all substrings of the text that are at a small distance from the pattern. The most popular distances are the Hamming distance, $L_1$, $L_2$, $L_{\infty}$-distances, and the edit distance. Unfortunately, the general task of computing these distances for each alignment of the pattern and of the text precisely requires at least $\Omega(m)$ space. This lower bound holds even for randomised algorithms~\cite{CJPS:2012}.

However, this is not the case if we allow approximation or if it is sufficient to compute the exact distances only when they are small. Here we focus on approximate pattern matching under the Hamming distance. Let $m$ be the length of the pattern and $n$ be the length of the text. If we are interested in computing a $(1+\eps)$-approximation of the Hamming distance for all alignments of the pattern and of the text, then this can be done in $\Oh(\eps^{-5} \sqrt{m} \log^{4} m)$ space and $\Oh(\eps^{-4} \log^3 {m})$ time per arriving symbol~\cite{CS:16}.
If we are only interested in computing the Hamming distances at the alignments where they do not exceed a given threshold $k$, which we call the $k$-\Mis problem, then Porat and Porat~\cite{Porat:09} showed a randomised streaming algorithm that solves this problem in $\Oh(k^3\log^7 m/\log\log m)$ space and $\Oh(k^2\log^5 m/\log\log m)$ time per arriving symbol; they also presented an algorithm for a special case of $k=1$ with $\Oh(\log^4 m / \log \log m)$ space and $\Oh(\log^3 m / \log \log m)$ time per arriving symbol. Their result was improved (in terms of the dependency on $k$) to $\Oh(k^2 \log^{11} m/\log \log m)$ space and $\Oh(\sqrt{k} \log k + \log^5 m)$ time per arriving symbol by Clifford et al.~\cite{k-mismatch}.\footnote{The logarithmic factors are hidden in~\cite{k-mismatch}, but can be restored easily from~\cite{Porat:09,k-mismatch}.} The error probability of all these solutions is at most $1/\poly(m)$.

Our first contribution is a new streaming algorithm for the $k$-\Mis problem. The crucial feature of our algorithm is that, for each alignment where the Hamming distance is at most~$k$, it can also output the differences of symbols of the pattern and of the text in the mismatching positions. This is particularly surprising as we are not allowed to store a copy of the pattern or of the text in the streaming setting. The $k$-\Mis problem extended with computing this additional characteristic is called here the $k$-\MwDR problem. We first develop a solution for $k=1$.

\begin{theorem}\label{th:1-mismatch}
$1$-\MwDR can be solved in $\Oh(\log^5 m / \log \log m)$ space and $\Oh(\log^5 m / \log \log m)$ time per arrival. The probability of error is at most~$1/\poly(n)$. 
\end{theorem}

As a corollary we obtain a $k$-\MwDR algorithm for arbitrary $k$ via an existing randomised reduction to the case of $k=1$~\cite{Porat:09,k-mismatch}.

\begin{theorem}\label{th:k-mismatch}
$k$-\MwDR can be solved in $\Oh(k^2 \log^{10} m/ \log \log m)$ space and $\Oh(k \log^8 m / \log \log m)$ time per arrival. The probability of error is at most~$1/\poly(m)$. 
\end{theorem}

A comparison of the complexities of our algorithm and the previous algorithms can be found in Table~\ref{tab:k-mis}. Since in the $k$-\MwDR problem we need at least $\Omega(k)$ time per arrival to list the symbol differences, the dependency of the running time of our algorithm on $k$ is optimal. The time complexity of our $k$-\MwDR algorithm is better than the time complexity of the $k$-\Mis algorithm of~\cite{Porat:09} and worse than that of~\cite{k-mismatch}. Our algorithm also uses less space than the $k$-\Mis algorithm of \cite{Porat:09} (in terms of $k$) and the $k$-\Mis algorithm of \cite{k-mismatch}. For the former, it is explained by the fact that we use a more efficient reduction. For the latter, it is because we do not need the involved time-saving techniques of Clifford et al.~\cite{k-mismatch} (as the running time of our algorithm is already almost optimal). We note that a similar problem was considered previously by Porat and Lipsky~\cite{kmis-errorcorrecting}. They introduced a small-space representation of a stream called a ``sketch'' that can be efficiently maintained under a symbol change and can be used to compute the Hamming distance between two streams as well as to correct the errors between them. Unfortunately, it is not clear whether their sketches can be efficiently maintained over a sliding window (i.e.\ for substrings of the text) and therefore we cannot use them in our model. 

\begin{table}[htpb]
  \begin{center}
  \small
    \begin{tabular}{|c|c|c|p{0.8cm}|}
      \hline
      Algorithm & Space & Time per symbol & \textsc{Err. Corr.} \\\hline\hline
      Porat and Porat~\cite{Porat:09} & $\Oh(k^3\log^7 m/\log\log m)$ & $\Oh(k^2\log^5 m/\log\log m)$ & No \\\hline
      Clifford et al.~\cite{k-mismatch} & $\Oh(k^2 \log^{11} m/\log \log m)$ & $\Oh(\sqrt{k} \log k + \log^5 m)$ & No \\\hline
      This paper & $\Oh(k^2 \log^{10} m/ \log \log m)$ & $\Oh(k \log^8 m / \log \log m)$ & Yes \\\hline
      Clifford et al.~\cite{DBLP:conf/soda/CliffordKP19} & $\Oh(k \log\frac{n}{k})$ & $\Oh(k+\log \frac{n}{k}(\sqrt{k \log k}+\log^3 n))$ & Yes \\\hline
    \end{tabular}
  \end{center}
  \caption{
    Comparison of the previous solutions to the $k$-\Mis problem with our solution to the $k$-\MwDR problem and a parallel contribution by Clifford et al.~\cite{DBLP:conf/soda/CliffordKP19}.
      \label{tab:k-mis}
  }
\end{table}

In parallel to our research, Clifford et al.~\cite{DBLP:conf/soda/CliffordKP19} presented a new streaming algorithm for the $k$-\MwDR problem with space complexity $\Oh(k \log\frac{n}{k})$ which spends $\Oh(k+\log \frac{n}{k}(\sqrt{k \log k}+\log^3 n))$ time per text symbol. Their algorithm introduces a new, much more elaborate way to sketch the pattern and the text. In particular, in contrast to the approach used by Porat and Porat~\cite{Porat:09} and Clifford et al.~\cite{k-mismatch} their algorithm does not reduce the $k$-\Mis problem to the $1$-\Mis problem, but solves the problem immediately for an arbitrary $k$. 

The \textsc{Error Correcting} feature is a powerful tool. We demonstrate this by using it to develop efficient streaming algorithms for the problem of pattern matching on weighted strings.
A \emph{weighted string} (also known as weighted sequence, position probability matrix or position weight matrix, PWM) is a sequence of probability distributions on the alphabet. Weighted strings are a commonly used representation of uncertain sequences in molecular biology. In particular, they are used to model motifs and have become an important part of many software tools for computational motif discovery; see e.g.\ \cite{JASPAR,Xia}. In the \WPM problem we are given a text and a pattern, both of which can be either weighted or regular strings. If either only the text or only the pattern are weighted, the task is to find all alignments of the text and of the pattern where they match with probability above a given threshold $1/z$. In the most general case, when both the pattern and the text is weighted, we must find all alignments of the text and of the pattern where there exists a regular string that matches both the text and the pattern with probability at least $1/z$. We assume here that the distributions at the respective positions are independent.

As previously, let $m$ be the length of the pattern, $n$ be the length of the text and assume a constant-sized alphabet. We are the first to consider the \WPM problem in the streaming setting. In the offline setting the most commonly studied variant, when the text is a weighted string and the pattern is a regular string, can be solved in $\Oh(n \log n)$ time via the Fast Fourier Transform~\cite{KCL_publication} or in $\Oh(n \log z)$ time using the suffix array and lookahead scoring \cite{Kociumaka2018}. This variant has been also considered in the indexing setting, in which we are to preprocess a weighted text to be able to answer queries for multiple patterns; see~\cite{amir_weighted_property_matching_j,DBLP:conf/cpm/BartonKPR16,DBLP:conf/edbt/BiswasPTS16,costas_weighted_suffix_tree_j}. The symmetric variant of the \WPM problem, when only the pattern is weighted, is closely related to the problem of profile matching \cite{DBLP:journals/tcs/PizziU08} and admits $\Oh(n \log n)$-time and $\Oh(n \log z)$-time solutions as well. Finally, the variant when both the text and the pattern are weighted was introduced in \cite{DBLP:journals/tcs/BartonLP16}, where an $\Oh(n z^2 \log z)$-time solution was presented. Later a more efficient $\Oh(n \sqrt{z} \log \log z)$-time solution was devised in \cite{Kociumaka2018}. The offline algorithms use $\Omega(m)$ space and the best indexing solution uses $\Omega(nz)$ space. A problem of computing Hamming and edit distances for weighted strings has been also considered~\cite{Amir2006}.


We consider each of the three variants of the \WPM problem.  If $z \ge m$, the offline algorithms listed above~\cite{Kociumaka2018,KCL_publication} and the black box transformation~\cite{blackboxonline} give $O(m)$-space on-line algorithms. The time complexities are $\Oh(\log^2 m)$ for the variants where only the pattern or only the text are weighted, and $\Oh(\sqrt{z} \log m \log \log z)$ for the variant when both the pattern and the text are weighted. Below we assume $z \le m$. 
For the two variants of \WPM where the text is weighted, our solutions are approximate. Namely, at each alignment we output either ``Yes'' or ``No''. If the pattern matches the fragment of the text, we output ``Yes''. If the match probability is between $(1-\eps)/z$ and $1/z$, we output either ``Yes'' or ``No'', and otherwise we output ``No''. If we output ``Yes'', we also output a $(1-\eps)$-approximation of the match probability between $P$ and $T$.

We show two series of streaming algorithms for the \WPM problem. The first one is based on the $k$-\MwDR problem. Let $\S_{\log z}$, $\T_{\log z}$, and $\PP_{\log z}$ be the space and the time complexities and error probability for the $k$-\MwDR for $k = \log z$, a pattern of length $m$, and a text of length $n$. With Theorem~\ref{th:k-mismatch} we obtain space $\S_{\log z} = \Oh(\log^2 z \cdot \log^{10} m / \log \log m)$, time $\T_{\log z} = \Oh(\log z \cdot \log^8 m / \log \log m)$, and error probability $\PP_{\log z}=1/\poly(m)$ (the new work of Clifford et al.~\cite{DBLP:conf/soda/CliffordKP19} yields space $\S_{\log z}=\Oh(\log z \log m)$, time $\T_{\log z}=\Oh(\log z + \log m(\sqrt{\log z \log \log z}+\log^3 m))$, and error probability $\PP_{\log z}=1/\poly(n)$). 

\begin{theorem}\label{th:WPM-kmismatch}
Assume that $z \le m$. If only the pattern is weighted, there is a streaming algorithm that solves the \WPM problem in $\Oh(z)+ \S_{\log z}$ space and $\Oh(\log^2 z)+\T_{\log z}$ time per arrival. If only the text is weighted, the problem can be solved $(1-\eps)$-approximately in $\Oh(z \log_{1/(1-\eps)} z)+ \S_{\log z}$ space and $\Oh(z \log_{1/(1-\eps)} z) + \T_{\log z}$ time.
At each arrival, the algorithms can err with probability~$\PP_{\log z}$. 
\end{theorem}

Second, we show three streaming algorithms for the \WPM problem that are based on the recent breakthroughs for the streaming \DM problem~\cite{DBLP:conf/esa/CliffordFPSS15, DBLP:conf/esa/GolanP17, DBLP:conf/icalp/GolanKP18}, the last two carried out in parallel with our research. In the \DM problem we are given a set of strings (a \emph{dictionary}) and a set of texts that arrive in a streaming fashion. When a new symbol of a text arrives, the algorithm must decide if the current text ends with an occurrence of a string in the dictionary. In the weighted-pattern-regular-text version of \WPM, we use a single-text-stream version of \DM for which the current best algorithm is by Golan and Porat~\cite{DBLP:conf/icalp/GolanKP18}, which for a dictionary $D$ of $m$-length strings takes $\Oh(|D| \log m)$ space and $\Oh(1)$ time per symbol. In the weighted-text-regular-pattern case we apply the \EPM algorithm of Breslauer and Galil~\cite{DBLP:journals/talg/BreslauerG14}. Finally, in the both-weighted case we apply the most general version of \DM for which the algorithm of Golan et al.~\cite{DBLP:conf/icalp/GolanKP18} uses $\Oh(|D| \log m)$ shared memory, $\Oh(\log |D| \log m)$ space per stream, and $\Oh(\log m)$ time per each arriving character of the text.

\begin{theorem}\label{th:WPM-dictionary}
Assume that $z \le m$. If only the pattern is weighted, there is a streaming algorithm that solves the \WPM problem in $\Oh(z \log m)$ space and $\Oh(1)$ time per arrival. If only the text is weighted, the problem can be solved $(1-\eps)$-approximately in $\Oh(z (\log_{1/(1-\eps)} z +\log m))$ space and $\Oh(z \log_{1/(1-\eps)} z)$ time. Finally, when both the pattern and the text are weighted, there is a $(1-\eps)$-approximation streaming algorithm with space complexity $\Oh(z(\log_{1/(1-\eps)} z + \log z \log m))$ that uses $\Oh(z (\log_{1/(1-\eps)} z + \log z\log m))$ time per arrival. At each arrival, the algorithms can err with probability~$1/\poly(n)$. 
\end{theorem}

Finally, in Section~\ref{sec:WPT-lb} we show that the space complexity of our algorithms is almost optimal.

\begin{proposition}\label{prp:WPT-lb}
If $z \le m$, then any streaming algorithm, exact or $(1-\eps)$-approximate, solving one of the three variants of the \WPM problem must use $\Omega(z)$ space. 
\end{proposition}

This is an extended version of a paper that was published at DCC 2017 conference~\cite{DBLP:conf/dcc/RadoszewskiS17}.

\paragraph{Model of computation} We assume that we receive the pattern first and can preprocess it by reading it in a streaming fashion several times. After having preprocessed the pattern we receive the text that arrives as a stream, one symbol at a time. We account for all the space that is used after the preprocessing and cannot afford to store a copy of the text or of the pattern. The indicated error probabilities are per arrival. We assume a constant-sized alphabet $\Sigma$. A symbol of a weighted string is a vector of probabilities of the letters in which all entries are of the form $c^{p/2^{dw}}$, where $w$ is the machine word size, $c$ and $d$ are constants, and $p$ is an integer that fits in a constant number of machine words (log-probability model). Additionally, the probability 0 has a special representation. The only operations on probabilities in our algorithms are multiplications and divisions, which can be performed exactly in $\Oh(1)$ time in this model.

\section{$k$-Mismatch with Error Correcting}
\label{sec:k-mismatch}
	In this section we give our solution to the $k$-\MwDR problem for $k=1$ (Theorem~\ref{th:1-mismatch}) and for a general value of $k$ (Theorem~\ref{th:k-mismatch}).

\subsection{Proof of Theorem~\ref{th:1-mismatch}: $k = 1$}
Let us start with a brief overview of the algorithm. Assume that letters of a string are numbered starting from 1. We reduce the $1$-\MwDR problem to $\Oh(\log m)$ instances of a special case of this problem where the mismatch is required to belong to the \emph{second half} of the pattern. More formally, consider $\ceil{\log m}+1$ partitions $P = P_i S_i$, where $P_i$ is a prefix of length $\min\{2^i, m\}$ and $S_i$ is the remaining suffix, for $i = 0, 1, \ldots, \ceil{\log m}$. We say that a substring of $T$ is a \emph{right-half $1$-mismatch occurrence} of $P_i$ if either $i = 0$ and $P_0$ does not match, or $i \ge 1$ and the mismatch is at position $j > 2^{i-1}$ in $P_i$.

\begin{observation}~\label{obs:special}
Any $1$-mismatch occurrence of $P$ in $T$ is a right-half $1$-mismatch occurrence of some $P_i$ followed by an exact occurrence of $S_i$.
\end{observation}

We run two parallel processes for each $i$. The first process searches for right-half $1$-mismatch occurrences of $P_i$. After having found a right-half $1$-mismatch occurrence, it passes the information about it to the second process that decides if it is followed by an exact occurrence of $S_i$. 

\subsubsection{Preliminaries}

\paragraph{Rabin--Karp fingerprints and 1-mismatch sketches}
The Rabin--Karp fingerprint of a string $X=X[1] \ldots X[\ell]$ is defined as $\phi(X) = (\sum_{i=1}^{\ell} X[i] \cdot r^i)  \bmod p$, where $p$ is a prime and $r$ is a random integer in $\mathbb{F}_p$. If we choose $p$ to be large enough, then the collision probability for any two $\ell$-length strings $X$ and $Y$, where $\ell \le m$, will be at most $1/\poly(m)$~\cite{KR:1987}. We will also need the following fact which follows immediately from the definition.

\begin{fact}\label{fct:addsubphi}
Let $X, Y$ be two strings and $Z = XY$ be their concatenation. Assuming that together with a Rabin--Karp fingerprint of a string of length $\ell$ we store $r^{\ell} \bmod p$ and $r^{-\ell} \bmod p$, from the Rabin--Karp fingerprints of two of the strings $X$, $Y$, $Z$ one can compute the Rabin--Karp fingerprint of the third string in $\Oh(1)$ time using the formula:
$$\phi(Z) = (\phi (X) + r^{|X|} \cdot \phi(Y)) \bmod p.$$
\end{fact}

We now introduce a notion of $1$-mismatch sketches and show its basic properties.

\begin{definition}[$1$-mismatch sketch]
For a prime $q$, the $1$-mismatch sketch of a string $X$ is a vector of length $q$, where the $j$-th element  is the Rabin--Karp fingerprint of the subsequence $X[j] X[j+q] X[j+2q]\ldots$
\end{definition}

This notion was implicitly used by Porat and Porat~\cite{Porat:09} in their $1$-\Mis algorithm.
In particular, they showed the following lemma.

\begin{lemma}[\cite{Porat:09}]\label{lem:Porat}
The Hamming distance between two $m$-length strings $X$ and $Y$ is equal to $1$ if and only if for each prime $q \in [\log m, 2 \log m]$ there is exactly one $j \le q$ such that subsequences $X[j] X[j+q] X[j+2q] \ldots$ and $Y[j] Y[j+q] Y[j+2q] \ldots$ are not equal. Moreover, knowing which subsequences are not equal, we can determine the mismatch position between $X$ and $Y$ in $\Oh(\log m)$ time. 
\end{lemma}

Let us list a few simple properties of 1-mismatch sketches that generalise Fact~\ref{fct:addsubphi}.

\begin{lemma}\label{lm:1-mismatch-sketches}
Let $X, Y$ be two strings and $Z = XY$ be their concatenation. Consider the $1$-mismatch sketches of $X$, $Y$, and $Z$ defined for a prime $q$. Then:
\begin{enumerate}[(i)]
\item\label{it:eq} If $X = Y$, then their $1$-mismatch sketches are equal;
\item\label{it:concat} Given the $1$-mismatch sketches of $X$ and $Y$, we can compute the $1$-mismatch sketch of $Z$ in $\Oh(q)$ time;
\item\label{it:subtract} Given the $1$-mismatch sketches of $X$ and $Z$, we can compute the $1$-mismatch sketch of $Y$ in $\Oh(q)$ time;
\item\label{it:power} Given the $1$-mismatch sketch of $X$, we can compute the $1$-mismatch sketch of $X^\alpha$ in $\Oh(q \log \alpha)$ time, where $\alpha$ is an integer  and $X^\alpha$ is a concatenation of $\alpha$ copies of $X$.
\end{enumerate}
\end{lemma}
\begin{proof}
Property~\eqref{it:eq} is a direct corollary of the definitions of Rabin--Karp fingerprints and $1$-mismatch sketches. As for Property~\eqref{it:concat}, note that we need to compute the Rabin--Karp fingerprints of the at most $q$ concatenations of pairs of strings. Similarly, for Property~\eqref{it:subtract} we only need to compute the Rabin--Karp fingerprints of the at most $q$ strings constructed from $Y$ given the Rabin--Karp fingerprints of the at most $q$ strings constructed from $X$ and their concatenations (in $Z$). Thus, Properties~\eqref{it:concat} and \eqref{it:subtract} follow from Fact~\ref{fct:addsubphi}. Finally, Property~\eqref{it:power} is implied by Property~\eqref{it:concat} as we can compute the $1$-mismatch sketch of a square of any string given its $1$-mismatch sketch in $\Oh(q)$ time.  
\end{proof}

\paragraph{Periodicity in strings} 
For a string $Q$, by $Q[i, j]$ we denote a substring of $Q$ equal to $Q[i]\ldots Q[j]$.
We say that a string $Q$ has a \emph{period} $p$ if $Q[i]=Q[i+p]$ for all $i=1,\ldots,|Q|-p$ (equivalently, if $Q[1,|Q|-p]=Q[p+1,|Q|]$). If $Q$ has a period $p$, then the string $Q[1,p]$ is the corresponding \emph{string period} of $Q$. 

\begin{lemma}[Fine and Wilf's periodicity lemma~\cite{fw:65}]\label{lem:fine_wilf}
If $Q$ has two periods $p$ and $q$ such that $p+q \le |Q|$, then $Q$ also has a period $\gcd(p,q)$. 
\end{lemma}

We will need the following two corollaries of this lemma. First, all the periods of $Q$ not greater than the half of the length of $Q$ are multiplicities of the shortest period of $Q$. And secondly, if $Q$ has at least 3 occurrences in a string $P$ such that $|P| \le 2|Q|$, then all the occurrences of $Q$ in $P$ form an arithmetic progression with the difference equal to the shortest period of $Q$.

\paragraph{Streaming algorithm for \EPM} 
Let us now recall a streaming \EPM algorithm~\cite{Porat:09,PeriodicityInStreams,DBLP:journals/talg/BreslauerG14} for a pattern $Q$ and text $T$ that uses $\Oh(\log |Q|)$ space and takes $\Oh(\log |Q|)$ time per symbol. The algorithm stores $\Oh(\log |Q|)$ levels of positions of $T$. Positions in level $j$ are occurrences of $Q[1,2^j]$ in the suffix of the current text $T$ of length $2^{j+1}$. The algorithm stores the Rabin--Karp fingerprints of strings from the beginning of $T$ up to each of these positions. If there are at least 3 such positions at one level, then, by Lemma~\ref{lem:fine_wilf}, all the positions form a single arithmetic progression whose difference equals the length of the minimal period of $Q[1,2^j]$. This allows to store the aforementioned information very compactly, using only $\Oh(\log |Q|)$ space in total. Finally, the algorithm stores the Rabin--Karp fingerprint of the current text and of all prefixes $Q[1,2^j]$. When a new symbol $T[q]$ arrives, the algorithm considers the leftmost position $\ell_j$ in each level $j$. If $q-\ell_j+1$ is smaller than $2^{j+1}$, the algorithm does nothing. Otherwise if the fingerprints imply that $\ell_j$ is an occurrence of $Q[1,2^{j+1}]$, the algorithm promotes it to the next level, and if $\ell_j$ is not an occurrence of $Q[1,2^{j+1}]$, the algorithm discards it. When a position reaches the top level, it is an occurrence of $Q$ and the algorithm outputs it. We note that the algorithm above is not the fastest known solution, but it gives the desired time bounds and is conceptually simple. For more details, see~\cite{Porat:09,PeriodicityInStreams,DBLP:journals/talg/BreslauerG14}. 

\medskip
We now describe our 1-\MwDR algorithm in detail. 

\subsubsection{Right-half $1$-mismatch occurrences of $P_i$} 
We start by describing the first process that detects right-half $1$-mismatch occurrences of $P_i$. When the process finds a new occurrence $T[q-|P_i|+1, q]$, it also computes the $1$-mismatch sketches of $T[q-|P_i|+1, q]$ and sends them to the second process. 

Suppose a new symbol $T[q]$ has arrived. If $i = 0$, $P_i = P[1]$ is a single letter string. The process compares $T[q]$ and $P[1]$, and if they are not equal, outputs $q$ and $1$-mismatch sketches of $T[q]$. 

Suppose now that $i > 0$. During the preprocessing phase, we compute the $1$-mismatch sketches of $P_i$ for all primes in $[\log m, 2 \log m]$. During the main phase, we run two algorithms for $P_i$ and $T$: the $1$-\Mis algorithm~\cite{Porat:09} and the \EPM algorithm. The $1$-\Mis algorithm~\cite{Porat:09} identifies  $1$-mismatch occurrences of $P_i$ in $T$, and for each of them returns the mismatch position. The algorithm uses $\Oh(\log^4 m / \log \log m)$ space and $\Oh(\log^3 m / \log \log m)$ time per symbol. However, the algorithm does not know the difference of symbols at the mismatch position, and we will need this information in order to compute the $1$-mismatch sketches. To extend the algorithm with the required functionality, we use the \EPM algorithm. If $T[q-|P_i|+1, q]$ is a right-half $1$-mismatch occurrence of $P_i$ for $i>0$, then $P_i[1,2^{i-1}]$ matches at the position $q-|P_i|+1$ exactly. It follows that at time $q$ the position $q-|P_i|+1$ is stored at level $i-1$ of the \EPM algorithm, and we know the Rabin--Karp fingerprints of $T[1, q-|P_i|]$ and $T[1,q]$. Therefore, we can compute the Rabin--Karp fingerprint of $T[q-|P_i|+1,q]$ in $\Oh(1)$ time using Fact~\ref{fct:addsubphi}. Let $j$ be the mismatch position and $\phi (P_i)$ and  $\phi (T[q-|P_i|+1,q])$ be the Rabin--Karp fingerprints of $P_i$ and $T[q-|P_i|+1,q]$, respectively. We use the fingerprints to compute the difference of symbols of $P_i$ and $T[q-|P_i|+1,q]$ at position $j$.

\begin{lemma}\label{lem:DataRecovery}
Assume that $X$ and $Y$ are two strings of length $m$ that differ only at position $j$. Knowing the Rabin--Karp fingerprints of $X$ and $Y$, we can compute $X[j]-Y[j]$ in $\Oh(\log m)$ time.
\end{lemma}
\begin{proof}
Let $\phi(X)$ and $\phi(Y)$ be the Rabin--Karp fingerprints of $X$ and $Y$, respectively. Then $X[j]-Y[j]$ is equal to $(\phi(X) - \phi(Y)) \cdot r^{-j} \pmod p$, where $p$ and $r$ are the integers used in the definition of Rabin--Karp fingerprints. Finally, $r^{-j} \pmod p$ can be computed in $\Oh(\log m)$ time.
\end{proof}

Now that we know the mismatch position and the letter difference, we can compute the $1$-mismatch sketches of $T[q-|P_i|+1,q]$ in $\Oh(\log^2 m / \log \log m)$  time from the 1-mismatch sketches of $P_i$.

\subsubsection{Exact occurrences of $S_i$} 
If $i=\ceil{\log m}$, $S_i$ is the empty string and therefore the second process is not necessary. Henceforth we assume $0 \le i < \ceil{\log m}$. The second process is built on top of the \EPM algorithm for the pattern $Q = S_i$ and $T$. Since for each new position $q$ the first process tells whether it is preceded by a right-half $1$-mismatch occurrence of $P_i$, all we need is to carry this information from the level $0$ of the \EPM algorithm to the top level. We claim that it suffices to store the $1$-mismatch sketches for a constant number of positions in each level. Using them, we will be able to infer the remaining unstored information.

Consider level $j$. Let us assume that the algorithm has read $T[1,q]$ so far. The progression that the \EPM algorithm currently stores for this level can be a part of a longer progression of occurrences of $S_i[1,2^j]$ in~$T$. More precisely, we consider the maximal arithmetic progression $R_j$ of occurrences of $S_i[1,2^j]$ in $T[1,q]$ with difference $\rho_{ij}$ being the shortest period of $S_i[1,2^j]$. Then, if the \EPM algorithm stores at least three occurrences, they form a suffix of the progression $R_j$. Otherwise there are at most two occurrences stored, so only the first occurrence may not belong to $R_j$ but to a previous such maximal progression.

Let $\ell$ be some occurrence of $S_i[1,2^j]$ for which we would like to figure out whether it is preceded by a right-half $1$-mismatch occurrence of $P_i$. If $\ell$ is far from the start of the current progression $R_j$, then the text preceding $\ell$ is periodic with period $\rho_{ij}$ and we can use this fact to infer the $1$-mismatch information. So our main concern is the positions $\ell$ that are at the distance of at most $|P_i| = 2^i$ from the start of the progression $R_j$. We define four positions $\ell_j^a$, $a = 1,2,3,4$, relative to the progression $R_j$ that help us to restore the information in this case. Note that we can easily determine the moment when a new progression $R_j$ starts, as this is precisely the moment when the difference between two consecutive positions in level $j$ becomes greater than $\rho_{ij}$. We define $\ell_j^1$ as the first position preceded by a right-half $1$-mismatch occurrence of $P_i$ that was added to level $j$ since that moment. We further define $\ell_j^a$ as the leftmost terms in $R_j$ preceded by a right-half $1$-mismatch occurrence of $P_i$ in $[\ell_j^1+ (a-1)\cdot 2^{i-2}, \ell_j^1+a \cdot 2^{i-2}]$ for $a = 2,3,4$. (If such terms exist; otherwise they are left undefined). A schematic view is given in Fig.~\ref{fig:ell_1234}. The algorithm stores $1$-mismatch sketches of $T[1, \ell_j^a-2^i-1]$ and $T[1, \ell_j^a-1]$ for each $a=1,2,3,4$.

\begin{figure}[htpb]
\begin{center}
\begin{tikzpicture}[scale=0.4]
\draw[thick] (-2,0)--(28,0);
\draw (-1,0) node[above] {$T$};
\foreach \x in {1,3,...,27}{\draw[thick,xshift=\x cm] (0,-0.1) -- (0,0.1);}
\draw[latex-latex] (1,2) -- node[above] {$\rho_{ij}$} (3,2);
\foreach \x/\c in {1/N,3/Y,5/Y,7/N,9/N,11/Y,13/Y,15/Y,17/Y,19/N,21/N,23/N,25/Y,27/N}{\draw[xshift=-0.35cm,yshift=-0.35cm] (\x,0) node {\scriptsize \c};}
\draw (3,1) -- (23,1);
\foreach \x in {3,8,...,23}{\draw[xshift=\x cm,yshift=1cm] (0,-0.1) -- (0,0.1);}
\foreach \x in {5.5,10.5,...,20.5}{\draw[xshift=\x cm,yshift=1cm] (0,0) node[above] {$2^{i-2}$};}
\draw (3,-2) node (A) {$\ell_j^1$};
\draw[-latex,densely dotted] (A) -- (3,-0.2);
\draw (11,-2) node (B) {$\ell_j^2$};
\draw[-latex,densely dotted] (B) -- (11,-0.2);
\draw (13,-2) node (C) {$\ell_j^3$};
\draw[-latex,densely dotted] (C) -- (13,-0.2);
\end{tikzpicture}
\end{center}
\caption{A progression $R_j$; the letter Y represents a $1$-mismatch occurrence of $P_i$ preceding a given position from $R_j$ and the letter N means there is no such occurrence. The special positions $\ell_j^1,\ell_j^2,\ell_j^3$ are marked whereas $\ell_j^4$ is undefined.}
\label{fig:ell_1234}
\end{figure}
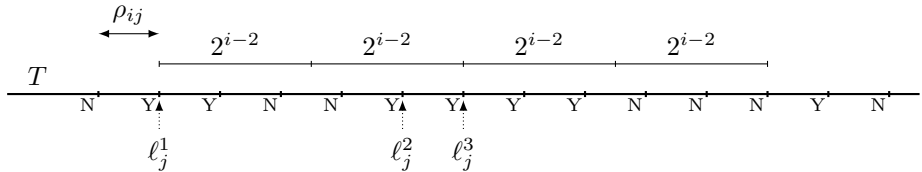

Let us mention that when the first element of a new progression $R_j$ arrives, the algorithm stores all the data for the previous progression $R_j$ as well. This accounts for the special case of just two occurrences of $S_i[1,2^j]$ in the suffix of $T[1,q]$ belonging to two different progressions. Afterwards the data for the previous progression can be safely discarded.

Besides, the algorithm stores a number of $1$-mismatch sketches for the pattern, which are computed during the preprocessing step and in total occupy $\Oh(\log^3 m / \log \log m)$ space (for a fixed $i$). First, it stores the $1$-mismatch sketches of the shortest string period $P_i[1,\gamma_i]$ of $P_i[1,2^{i-1}]$. Second, it stores the $1$-mismatch sketches of the shortest string period $S_i[1, \rho_{ij}]$ of $S_i [1, 2^j]$ for each $j = 0,1, \ldots, \floor{\log |S_i|}$. Finally, the algorithm stores the $1$-mismatch sketches of $S_i [1,\delta_{ij}]$, where $\delta_{ij}$ is the remainder of $-2^i$ modulo $\rho_{ij}$. To compute the periods and the sketches during the preprocessing step we make three passes over the pattern. Over the first two passes we compute the minimal periods of $P_i[1,2^{i-1}]$ and $S_i [1, 2^j]$ for all $i$ and $j$ using $\Oh(\log^2 m)$ instances of the streaming algorithm~\cite{PeriodicityInStreams}. The algorithm requires $O(\log^2 m)$ space and computes the minimal period in one pass if the period is smaller than $m/2$, and in two passes if otherwise. During the third pass we compute the $1$-mismatch sketches.

\begin{lemma}\label{lm:occ_verification}
The algorithm can decide if a position $\ell \in R_j$ is preceded by a right-half $1$-mismatch occurrence of $P_i$ in $\Oh (\log^3 m / \log \log m)$ time and, if this is the case, it can also compute the $1$-mismatch sketches of $T[1,\ell-2^i-1]$ and $T[1,\ell-1]$.
\end{lemma}
\begin{proof}
We consider three cases based on the relationship between $\ell$ and the positions $\ell_j^a$.

\paragraph{Case 1: $\ell < \ell_j^1$ or $\ell_j^a+2^{i-2} \le \ell < \ell_j^{a+1}$ for some $a \in \{1,2,3\}$} In this case $\ell$ cannot be preceded by a $1$-mismatch occurrence of $P_i$ by definition. If $\ell_j^{a+1}$ is undefined, as the upper bound on $\ell$ we take $\ell_j^b$ for the smallest $b>a+1$ that is defined or, if no such $b>a+1$ exists, $\ell_j^1+2^i$.

\paragraph{Case 2: $\ell_j^a \le \ell < \ell_j^a + 2^{i-2}$ for some $a \in \{1,2,3,4\}$} Recall that $\gamma_i$ is the minimal period of $P_i[1,2^{i-1}]$. Let us first show that if $\ell$ is preceded by a right-half $1$-mismatch occurrence of $P_i$, then it is also preceded by a $1$-mismatch occurrence of $(P_i[1,\gamma_i])^\alpha P_i$, where $\alpha = (\ell - \ell_j^a) / \gamma_i$. To this end, recall that $\ell_j^a$ is preceded by a right-half $1$-mismatch occurrence of $P_i$. If the same property holds for $\ell$, then $T[\ell_j^a-2^i, \ell_j^a-2^{i-1}-1]$ and $T[\ell-2^i, \ell-2^{i-1}-1]$ are two occurrences of $P_i[1,2^{i-1}]$ that overlap by at least $2^{i-2}$ positions; see Fig.~\ref{fig:case2}. Therefore $T[\ell_j^a-2^i, \ell-2^i-1]$ is a power of the minimal period of $P_i[1,\gamma_i]$ by Lemma~\ref{lem:fine_wilf}.

Below we explain how to compute the $1$-mismatch sketches of $T[\ell_j^a-2^i, \ell-1]$ and of $(P_i[1,\gamma_i])^\alpha P_i$ in $\Oh (\log^3 m / \log \log m)$ time. The algorithm can then use Lemma~\ref{lem:Porat} to determine in $\Oh (\log^2 m / \log \log m)$ time whether $\ell$ is preceded by a $1$-mismatch occurrence of $(P_i[1,\gamma_i])^\alpha P_i$ and, if so, to determine the mismatch position. Hence, it can check whether $\ell$ is preceded by a right-half $1$-mismatch occurrence of $P_i$. Moreover, if this is the case, the algorithm can compute the $1$-mismatch sketch of $T[1,\ell-2^i-1]$ using sketches of $T[1, \ell_j^a-2^i-1]$ and $(P_i[1,\gamma_i])^\alpha$.

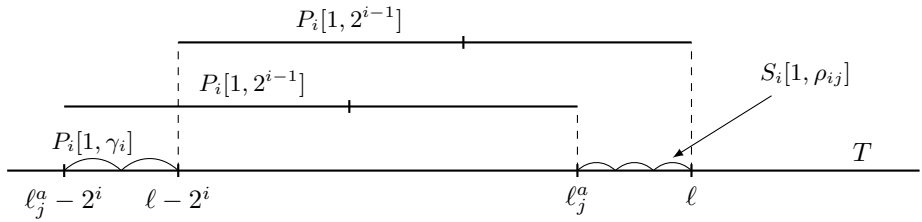
\begin{figure}[htpb]
\begin{center}
\begin{tikzpicture}[xscale=0.75,yscale=0.85]
\draw[thick] (0,0)--(16,0);
\draw[thick] (1,-0.1)--(1,0.1);
\draw[thick] (3,-0.1)--(3,0.1);
\draw[thick] (10,-0.1)--(10,0.1);
\draw[thick] (12,-0.1)--(12,0.1);
\node[below] at (1,-0.1) {$\ell_j^a-2^i$};
\node[below] at (3,-0.1) {$\ell-2^i$};
\node[below] at (10,-0.1) {$\ell_j^a$};
\node[below] at (12,-0.1) {$\ell$};
\draw (1,0) to[out=40, in=140] (2,0);
\draw (1.5,0.1) node[above] {\small $P_i[1,\gamma_i]$};
\draw (2,0) to[out=40, in=140] (3,0);
\draw[thick] (1,1)--(10,1);
\draw[dashed] (10,0)--(10,1);
\draw[thick] (6,0.9)--(6,1.1);
\node[above] at (4.3,1) {\small $P_i[1,2^{i-1}]$};
\draw[thick] (3,2)--(12,2);
\draw[dashed] (3,0)--(3,2);
\draw[dashed] (12,0)--(12,2);
\draw[thick] (8,1.9)--(8,2.1);
\node[above] at (6,2) {\small $P_i[1,2^{i-1}]$};
\draw (15,0) node[above] {$T$};
\draw (10,0) to[out=40, in=140] (10.666666,0);
\draw (10.666666,0) to[out=40, in=140] (11.333333,0);
\draw (11.333333,0) to[out=40, in=140] (12,0);
\draw (14,1.5) node (X) {\small $S_i[1,\rho_{ij}]$};
\draw[-latex] (X) -- (11.66666,0.2);
\end{tikzpicture}
\end{center}
\caption{Case 2 of Lemma~\ref{lm:occ_verification}. Position $\ell$ is close to one of the positions~$\ell_j^a$, which implies that $T[\ell_j^a-2^i, \ell-2^i-1]$ is a power of $P_i[1,\gamma_i]$.}
\label{fig:case2}
\end{figure}

Recall that the algorithm stores the $1$-mismatch sketches of $T[1,\ell_j^a-1]$. The algorithm can compute the $1$-mismatch sketches of $T[1,\ell-1]$ using the $1$-mismatch sketches of $T[1,\ell_j^a-1]$ and $S_i[1,\rho_{ij}]$ in $\Oh(\log^3 m / \log \log m)$ time via Lemma~\ref{lm:1-mismatch-sketches}\eqref{it:power} and Lemma~\ref{lm:1-mismatch-sketches}\eqref{it:concat}. Finally, the algorithm computes the $1$-mismatch sketches of $T[\ell_j^a-2^i, \ell-1]$ in $\Oh (\log^2 m / \log \log m)$ time by applying Lemma~\ref{lm:1-mismatch-sketches}\eqref{it:subtract} for the $1$-mismatch sketches of $T[1,\ell_j^a-2^i-1]$ and $T[1,\ell-1]$. The $1$-mismatch sketches of $(P_i[1,\gamma_i])^\alpha P_i$ are computed from sketches of $P_i[1,\gamma_i]$ and $P_i$ using Lemma~\ref{lm:1-mismatch-sketches}\eqref{it:power} and Lemma~\ref{lm:1-mismatch-sketches}\eqref{it:concat}. 

\paragraph{Case 3: $\ell \ge \ell_j^1+2^i$} In this case $T[\ell-2^i,\ell-1]$ is a suffix of $T[\ell_j^1,\ell-1]$, which is a power of the minimal string period of $S_i[1,2^j]$, that is, $S_i[1, \rho_{ij}]$. As we know the $1$-mismatch sketches of $T[1, \ell_j^1-1]$, $S_i[1,\rho_{ij}]$, and $S_i [1,\delta_{ij}]$ (recall that $\delta_{ij}$ is the remainder of $-2^i$ modulo $\rho_{ij}$), the algorithm can use Lemma~\ref{lm:1-mismatch-sketches} to compute the $1$-mismatch sketches of $T[1,\ell-1]$ (by extending $T[1, \ell_j^1-1]$ with a power of $S_i[1,\rho_{ij}]$) and $T[1,\ell-2^i-1]$ (by subtracting the sketches of a power of $S_i[1,\rho_{ij}]$ of exponent $\ceil{2^i/\rho_{ij}}$ and adding the sketches of $S_i [1,\delta_{ij}]$) in $\Oh(\log^3 m / \log \log m)$ time; see Fig.~\ref{fig:case3}. It can then use them to determine whether $T[\ell-2^i, \ell-1]$ is a $1$-mismatch occurrence of $P_i$. 

\begin{figure}[htpb]
\begin{center}
\begin{tikzpicture}[xscale=0.75,yscale=0.85]
\draw[thick] (0,0)--(15,0);
\draw[thick] (1,-0.1)--(1,0.1);
\node[below] at (1,-0.1) {$\ell_j^1$};
\foreach \x in {0,1,...,10}{\draw[xshift=\x cm] (1,0) to[out=40, in=140] (2,0);}
\draw (14,0) node[above] {$T$};
\draw (1.5,0.1) node[above] {\small $S_i[1,\rho_{ij}]$};
\draw[thick] (12,-0.1)--(12,0.1);
\node[below] at (12,-0.1) {$\ell$};
\draw[latex-latex] (4.4,0.7) -- node[above] {$2^i$} (12,0.7);
\draw[densely dotted] (4.4,0) -- (4.4,0.7);
\draw[densely dotted] (12,0) -- (12,0.7);
\draw[snake=brace] (4.4,-0.5) -- node[below] {\small $S_i[1,\delta_{ij}]$} (4,-0.5);
\draw[densely dotted] (4.4,0) -- (4.4,-0.5);
\draw[densely dotted] (4,0) -- (4,-0.5);
\end{tikzpicture}
\end{center}
\caption{Case 3 of Lemma~\ref{lm:occ_verification}. Position $\ell$ is far enough from $\ell_j^1$.}
\label{fig:case3}
\end{figure}
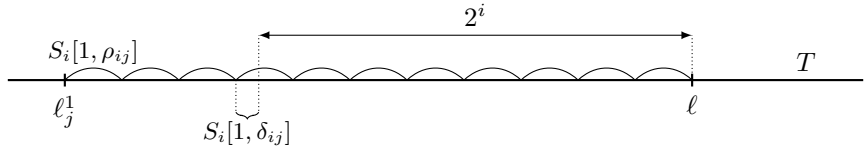

\end{proof}

The algorithm uses the lemma both to compile the output and to update the levels. When the algorithm encounters a position $\ell$ in the top level that is preceded by a right-half $1$-mismatch occurrence of $P_i$ and followed by an exact occurrence of $S_i$, the algorithm outputs it together with the required difference of symbols, which is computed from the 1-mismatch sketches of $T[\ell-2^i+1,\ell]$ using Lemma~\ref{lem:Porat} to determine the mismatch position and then Lemma~\ref{lem:DataRecovery}. Let us show how the algorithm updates the levels. When a new symbol $T[q]$ arrives and $T[q] = S_i[1]$, the algorithm adds $q$ to the level $0$. If we know from the first process that $q$ is preceded by a right-half $1$-mismatch occurrence of $P_i$, then the algorithm also tries to update the $\ell_0^a$ values and possibly retains the $1$-mismatch sketches of $T[1, q-2^i-1]$ and $T[1,q-1]$ output by the first process. The algorithm then updates each of the remaining levels in turn. To update level $j$, it considers the leftmost position $\ell_j$ in this level. If the Rabin--Karp fingerprints imply that it is an occurrence of $S_i[1,2^{j+1}]$, the algorithm promotes it to the next level, and otherwise discards it. In the former case the algorithm uses Lemma~\ref{lm:occ_verification} to check whether $\ell_j$ is preceded by a right-half $1$-mismatch occurrence of $P_i$. If it is, it computes the $1$-mismatch sketches and updates the positions $\ell_{j+1}^a$.

\subsubsection{Complexity} 
For each $i$ we run the two processes in parallel. The bottleneck of the first process is the $1$-\Mis algorithm; it uses $\Oh(\log^4 m / \log \log m)$ space and $\Oh(\log^3 m / \log \log m)$ time per symbol. The time complexity of the second process is bounded by $\Oh(\log m)$ applications of the algorithm of Lemma~\ref{lm:occ_verification} (one per level) and is $\Oh(\log^4 m / \log \log m)$ per symbol. The space complexity of the second process is $\Oh(\log^3 m / \log \log m)$. Therefore, both the space and the time that our $1$-\MwDR algorithm uses per symbol is $\Oh(\log^5 m / \log \log m)$.

\subsection{Proof of Theorem~\ref{th:k-mismatch}: general value of $k$}
Theorem~\ref{th:k-mismatch} follows from Theorem~\ref{th:1-mismatch} via a randomised reduction. The first variant of this reduction, which was deterministic, was presented by Porat and Porat~\cite{Porat:09}, who used it to reduce the $k$-\Mis problem to the $1$-\Mis problem. It was further made more space-efficient at a return of slightly higher error probability by Clifford et al.~\cite{k-mismatch}. The main idea of this reduction is to consider a number of partitions of the pattern into $\Oh(k \log^2 m)$ subpatterns. By defining the partitions appropriately, we can guarantee that at each alignment where the Hamming distance is small, each mismatch will correspond to a 1-mismatch occurrence of some subpattern. This lets us apply the $1$-\MwDR algorithm to find all such alignments and to restore the data for them. We give a full description of the reduction below. 

We start by filtering out the locations where the Hamming distance is large. To this end, we use a randomised algorithm for the following \emph{$2$-approximate $k$-mismatch problem}. Let $H$ be the true Hamming distance at a particular alignment of the pattern and the text. The algorithm outputs ``Yes'' if $H \le k$, either ``Yes'' or ``No'' if $k < H \le 2k$, and ``No'' if $H > 2k$. 

\begin{lemma}[\cite{k-mismatch}]\label{th:approximate}
Given a pattern $P$ of length $m$ and a streaming text arriving one symbol at a time, there is a randomised $\Oh(k^2 \log^6 m)$-space algorithm which takes $\Oh(\log^5 m)$ worst-case time per arriving symbol and solves the $2$-approximate $k$-mismatch problem. The probability of error is at most~$1/m^2$.
\end{lemma}

Next, we consider $\ell=\floor{\log m}$ partitions of the pattern into equispaced subpatterns. More precisely, we select $\ell$ random primes from the interval $[k \log^2 m, 34 k \log^2 m]$. For each prime $p_i$ the pattern $P$ is then partitioned into $p_i$ subpatterns $P_{i,j} = P [j] P [p_i+j] P[2p_i+j] \ldots$, where $j = 1,\ldots, p_i$. Consider a particular location $q$ in the text and the set of all mismatches between $P$ and $T[q-m+1,q]$. Let us call a mismatch \emph{isolated} if it is the only mismatch between some subpattern $P_{i,j}$ and the corresponding subsequence of $T[q-m+1,q]$, and let $M_q$ be the set of all isolated mismatches at an alignment $q$. 
 
\begin{lemma}[\cite{k-mismatch}]
If $|M_q| \le 2k$, the set $M_q$ contains all mismatches between $P$ and $T[q-m+1,q]$ with probability at least $1-1/m^2$.
\end{lemma}

To finalise the reduction, we partition the text $T$ into $p_i$ equispaced substreams $T_{i,j}$, where $j = 1,\ldots, p_i$, for each prime $p_i$. We run the $1$-\MwDR algorithm for all $\sum_i p_i^2 = \Oh(k^2 \log^5 m)$ subpattern/substream pairs $(P_{i,{j_1}}, T_{i,{j_2}})$. At each location where the algorithm for the $2$-approximate problem outputs ``Yes'', we retrieve all isolated mismatches and the data using the appropriate $p_i$ instances of the $1$-\MwDR problem for each $i$. This completes the description of our solution to the $k$-\MwDR problem. 

In total, we use $\Oh(k^2 \log^{10} m / \log \log m)$ space. To analyse the time complexity, note that when a new symbol of $T$ arrives, we need to send it to $\Oh(\log m)$ substreams only (one for each prime) and run the next step for each of the $\Oh(k \log^2 m)$ instances of the $1$-\MwDR problem on each of these substreams, which requires $\Oh(k \log^8 m / \log \log m)$ time per symbol.

\section{Applications~--- Weighted Pattern Matching}
\label{sec:WPM}
	Below we present the proofs of Theorems~\ref{th:WPM-kmismatch} and~\ref{th:WPM-dictionary} that give two series of streaming algorithms for the \WPM problem. We present the proofs in parallel. We conclude in Section~\ref{subsec:prop} with a proof of Proposition~\ref{prp:WPT-lb} that states an $\Omega(z)$ lower bound on the space complexities of such algorithms.
 
\subsection{Preliminaries} 
We first introduce a notion of a heavy string $\heavy{P}$  of a weighted string $P$.

\begin{definition}
 For a weighted string $P$, by $\heavy{P}$ we denote a regular string obtained from $P$ by choosing at each position the symbol with the maximum probability (ties are handled arbitrarily).
\end{definition}

The observation below, initially introduced in \cite{Kociumaka2018}, shows a key property of this notation.

\begin{observation}\label{obs:logz}
If a string $S$ matches a weighted string $P$ with probability at least $1/z$, then the number of mismatches between $\heavy{P}$ and $S$ is at most $\log z$.
\end{observation}
\begin{proof}
It follows from the fact that at each mismatch position the probability of $P$ and $S$ to match is at most $1/2$.
\end{proof}

\begin{observation}\label{obs:D}
The total number of strings that match a weighted string $P$ with probability at least $1/z$ is at most $z$.
\end{observation}
\begin{proof}
The sum of the match probabilities over all such strings cannot exceed~$1$. The claim follows.
\end{proof}

\subsection{Case 1~--- only pattern is weighted}
\label{sec:WPST}
We start by presenting two algorithms for the \WPM problem for a weighted pattern and a regular text.

\subsubsection{Solution via $k$-\MwDR} 
Alongside Observation~\ref{obs:logz}, the main idea of our solution is to find all $\log z$-occurrences of $\heavy{P}$ in $T$, and then to filter out those corresponding to the alignments where the match probability is too small. 

Let $D$ be the set of all (regular) strings that match the pattern with probability at least $1/z$. If a substring of the text $T$ matches $P$ with probability $\ge 1/z$, it must belong to the set $D$. During the preprocessing phase, we do not compute the set $D$ itself, but would like to compute a set $M$ of mismatches between $\heavy{P}$ and the strings in $D$. More formally:
$$M=\{(i,a)\,:\, a = S[i] \ne \heavy{P}[i],\,S \in D\}.$$
\begin{observation}\label{obs:M}
  $|M| \le z-1$.
\end{observation}
\begin{proof}
  From each mismatch $(i, a) \in M$ one can produce a regular string that matches $P$ with probability $\ge 1/z$ by taking $\heavy{P}$ and replacing $\heavy{P}[i]$ by $a$. Due to Observation~\ref{obs:D} and the fact that $\heavy{P} \in D$, there are at most $z-1$ such strings in $D$.
\end{proof}

\begin{example}\label{ex:WS}
  Consider the weighted pattern $P$ corresponding to the weighted string $X$ from Table~\ref{tab:WS}. We have $\heavy{P}=\mathtt{ABAB}$. Let $z=8$.

  \begin{table}[htpb]
      \begin{center}
      \begin{tabular}{r|p{0.4cm}|p{0.4cm}|p{0.4cm}|p{0.4cm}}
        $i$ &1&2&3&4\\\hline
        probability of $\mathtt{A}$ & $\tfrac12$ & 0 & $\tfrac12$ & $\tfrac16$ \\\hline
        probability of $\mathtt{B}$ & $\tfrac18$ & 1 & $\tfrac38$ & $\tfrac23$ \\\hline
        probability of $\mathtt{C}$ & $\tfrac38$ & 0 & $\tfrac18$ & $\tfrac16$ \\
      \end{tabular}
      \end{center}
    \caption{
      A weighted string $X$ of length 4 over $\Sigma=\{\mathtt{A},\mathtt{B},\mathtt{C}\}$.
        \label{tab:WS}
    }
  \end{table}

  We have $D=\{\mathtt{ABAB},\mathtt{AB\underline{B}B},\mathtt{\underline{C}BAB}\}$ with the probabilities $\tfrac16$, $\tfrac18$, and $\tfrac18$, respectively. Thus $M=\{(1,\mathtt{C}),(3,\mathtt{B})\}$ (the mismatches are underlined).
\end{example}

  To perform the preprocessing step in a streaming fashion with small space usage, we compute a slightly larger set $M'$ that consists of $z-1$ pairs $(i,a)$ for $i \in \{1,\ldots,m\}$ and $a \in \Sigma$ with the greatest value of $p(i,a)=\frac{\prob{P[i] = a}}{\prob{P[i] = \heavy{P}[i]}}$ (if there are less than $z-1$ such pairs in total, the set $M'$ contains all such pairs).
  
\begin{observation}\label{obs:M1}
  $M \subseteq M'$.
\end{observation}
\begin{proof}
  Assume to the contrary that $(i,a) \in M \setminus M'$. This implies that $|M'|=\floor{z}-1$. Using the same argument as in the proof of Observation~\ref{obs:M}, from every pair in $M' \cup (i,a)$ one can obtain a different string from $D$. This would yield $|D|>z-1$, which contradicts Observation~\ref{obs:D}.
\end{proof}

\begin{example}
For the weighted pattern from Example~\ref{ex:WS}, one might have
$$M'=\{
(1,\mathtt{C},\tfrac34),\,
(3,\mathtt{B},\tfrac34),\,
(1,\mathtt{B},\tfrac14),\,
(3,\mathtt{C},\tfrac14),\,
(4,\mathtt{A},\tfrac14),\,
(4,\mathtt{C},\tfrac14),\,
(2,\mathtt{A},0)
\}.$$
The elements of $M'$ are ordered by $p(i,a)$.
\end{example}
  
Each element of $M'$ is stored as a triple $(i,\heavy{P}[i]-a,p(i,a))$ in a priority queue ordered by $p(i,a)$. To construct the set $M'$, for $i=1,\ldots,m$ we insert all $(i,a)$ with $a \ne \heavy{P}[i]$ into the priority queue. We keep the size of $M'$ not greater than $z-1$ by removing the elements with the smallest $p(i,a)$ if needed. In the end we compute a balanced binary search tree that indexes all the mismatches from $M'$ with pairs $(i,\heavy{P}[i]-a)$. We also compute and store the probability $\pi$ that $\heavy{P}$ matches $P$. The time needed for the preprocessing step is $\Oh(z \log z)$ and the total space consumption is $\Oh(z)$.

During the main phase, we run a $k$-\MwDR algorithm for $k = \log z$, $\heavy{P}$, and $T$. From Observation~\ref{obs:logz} it follows that the algorithm will report all alignments where $P$ and $T$ match with probability at least $1/z$, but it might report some other alignments, and we need to filter them out. Recall that every occurrence found by the $k$-\MwDR algorithm applied for $\heavy{P}$ and $T$ with $k = \log z$ is reported together with the at most $k$ mismatch positions and the corresponding letter differences. Therefore, we can use the balanced binary search tree that stores the set $M'$ to compute the match probability by updating $\pi$ with the probabilities of the at most $k=\log z$ mismatches. If any of the mismatches is not present in $M'$, then by Observation~\ref{obs:M1} the candidate does not belong to the set $D$ and its match probability is certainly below the $1/z$ threshold. Thus the time needed to retrieve the match probability (provided that it is at least $1/z$) is $\Oh(\log^2 z)$. With this we arrive at the first algorithm of Theorem~\ref{th:WPM-kmismatch}:

\begin{proposition}
If only the pattern is weighted, there is a streaming algorithm that solves the \WPM problem in $\Oh(z)+ \S_{\log z}$ space and $\Oh(\log^2 z)+\T_{\log z}$ time per arrival. At each arrival, the algorithm can err with probability~$\PP_{\log z}$. 
\end{proposition}

Let us note that, in the considered case of $z < m$, $\T_{\log z} = \Omega(\log^2 z)$ both for the algorithm of Theorem~\ref{th:k-mismatch} and the algorithm of Clifford et at.~\cite{DBLP:conf/soda/CliffordKP19}.

\subsubsection{Solution via \DM} 
As in the previous subsection, let $D$ be the set of all (regular) strings that match the pattern with probability at least $1/z$. Let us recall that if a substring of the text $T$ matches $P$ with probability $\ge 1/z$, it must belong to the set $D$. To identify such substrings, we will use the streaming \DM algorithm for the set $D$ and the text~$T$. The current best algorithm is by Golan and Porat~\cite{DBLP:conf/icalp/GolanKP18}, which for a dictionary $D$ of $m$-length strings takes $\Oh(|D| \log m)$ space and $\Oh(1)$ time per symbol. The algorithm is randomised and has error probability $1/\poly(n)$.

We arrive at the first algorithm of Theorem~\ref{th:WPM-dictionary}:

\begin{proposition}
If only the pattern is weighted, there is a streaming algorithm that solves the \WPM problem in $\Oh(z \log m)$ space and $\Oh(1)$ time per arrival. At each arrival, the algorithm can err with probability~$1/\poly(n)$. 
\end{proposition}


\subsection{Case 2~--- only text is weighted}	
\label{sec:SPWT}
In this section we show two $(1-\eps)$-approximate solutions to the \WPM problem for a regular pattern $P$ and a weighted text~$T$. We assume that $\eps < 1/2$. If $\eps > 1/2$, we can use the $1/2$-approximation algorithm that has the same asymptotic complexity and better approximation factor. We start by giving a definition of maximal matching suffixes that will play a crucial role in both algorithms.

\begin{definition}
A \emph{maximal matching suffix} of a weighted string $T$ is a (regular) string $S$ such that $S$ matches $T[|T|-|S|+1,|T|]$ with probability at least $1/(2z)$ and either $|S|=|T|$ or any string $aS$, for $a \in \Sigma$, matches $T[|T|-|S|,|T|]$ with probability smaller than $1/(2z)$.
\end{definition}

\begin{remark}
  The reason for selecting the cut-off value of $1/(2z)$ will become clear later (in Lemma~\ref{lem:SWP_str}, where we need the cut-off to be at most $(1-\eps)/z$). 
\end{remark}

We will use the observation below; intuitively, it follows from the fact that the sum of probabilities of maximal matching suffixes of a string does not exceed~1.

\begin{observation}[\cite{amir_weighted_property_matching_j}]\label{obs:Amir}
  A weighted string has at most $2z$ maximal weighted suffixes.
\end{observation}

\begin{example}\label{ex:max_s}
  For the weighted text $T$ corresponding to the weighted string $X$ from Table~\ref{tab:WS} and $z=8$, the set of maximal matching suffixes is:
  $$\{
  \mathtt{ABAB}\ (\tfrac16),\,
  \mathtt{AB\underline{B}B}\ (\tfrac18),\,
  \mathtt{\underline{C}BAB}\ (\tfrac18),\,
  \mathtt{\underline{C}B\underline{B}B}\ (\tfrac{3}{32}),\,
  \mathtt{B\underline{C}B}\ (\tfrac{1}{12}),\,
  \mathtt{BA\underline{A}}\ (\tfrac{1}{12}),\,
  \mathtt{BA\underline{C}}\ (\tfrac{1}{12}),\,
  \mathtt{B\underline{B}\underline{A}}\ (\tfrac{1}{16}),\,
  \mathtt{B\underline{B}\underline{C}}\ (\tfrac{1}{16})
  \}.$$
  Mismatches with $\heavy{T}$ are underlined and matching probabilities of each suffix are listed.
\end{example}

Importantly, if $P$ matches $T[q-m+1,q]$ with probability $\ge 1/z$, then there is a maximal matching suffix $S$ of $T[1,q]$ that ends with $P$. Moreover, the match probability between $P$ and $T[q-m+1,q]$ is equal to the match probability between the $m$-length suffix of $S$ and $T[q-m+1,q]$. To be able to compute the latter, we introduce a new problem which we refer to as the \SWP problem. In this problem we are given a stream of numbers from $[0,1]$ and an integer $m$. Each time a new number arrives, we must update and output the product of numbers in the $m$-length suffix of the stream. We will develop a $(1-\eps)$-approximation solution to the \SWP problem.

We summarize our two solutions as Algorithm~\ref{alg:SPWT}. The only difference between them is how we find the maximal matching suffixes that end with $P$. In the following subsections we explain each step in detail.

\begin{algorithm}
\begin{algorithmic}
\For{each new text symbol $T[q+1]$}
	\State Update the set of maximal matching suffixes
	\For {each maximal matching suffix $S$}
		\State Run the next step of \SWP
	\EndFor
	\If {there is a maximal matching suffix $S$ that ends with $P$}
		\State Use the output of \SWP for $S$ to compute a $(1-\eps)$-approximation $p$ of the match probability between $P$ and $T$
    \State If $p \ge (1-\eps)/z$, report an occurrence
	\EndIf
\EndFor
\end{algorithmic}
\caption{\WPM ~--- only text is weighted}
\label{alg:SPWT}
\end{algorithm}

\subsubsection{Computing maximal matching suffixes} We maintain each maximal matching suffix $S$ of $T[1,q]$ as a stream. The stream represents letters $s_i$, $i \le q$, that correspond to $S[i+|S|-q]$ and probabilities $x_i$ that are equal to the match probability between $s_i$ and $T[i]$. The letters $s_i$ for $1 \le i \le q-|S|$ can be arbitrary (they will be implied by the previous steps of the algorithm) and probabilities $x_i$ for $1 \le i \le q-|S|$ correspond to the matching probability of these letters and $T[i]$. Each of the streams will be stored using only $\Oh(\log_{1/(1-\eps)} z)$ space (in particular, we do not store all $s_i$ and $x_i$ explicitly).

We call a position $i$ in the stream a \emph{mismatch position} if $s_i \ne \heavy{T}[i]$. Let $r$ be the rightmost mismatch position in the stream such that $\prod_{i \ge r} x_i < 1/(2z)$. If such a position does not exist, we put $r = 0$. By Observation~\ref{obs:logz}, there are at most $\log z+1$ mismatch positions to the right of $r$. We index the stream by these mismatch positions and the differences between $s_i$ and $\heavy{T}[i]$ in these positions. We also store the product of the probabilities located to the right of each of these at most $\log z+1$ mismatch positions.

When a new text symbol $T[q+1]$ arrives, we first create $|\Sigma|$ copies of each stream. We then add $\prob{T[q+1] = b}$ for each $b \in \Sigma$ to the $b$-th copy of each stream and update the mismatches, the indices, and the  related information in $\Oh(\log z)$ time using the stored products of probabilities. At this moment the value $r$ in some streams might increase.
Furthermore, there might appear ``duplicate'' streams with equal indices. Duplicate streams correspond to a single maximal matching suffix and possibly its suffixes (which are not maximal matching suffixes of $T[1,q+1]$). We sort the streams by building a trie on their indices and delete the duplicates, leaving for each stream index one stream with the smallest position $r$. This takes $\Oh(z \log z)$ space and time in total.

\begin{example}
  Let us consider an extension of the weighted string from Table~\ref{tab:WS} by a single position that is shown in Table~\ref{tab:WS2}.

  \begin{table}[htpb]
      \begin{center}
      \begin{tabular}{r|p{0.4cm}|p{0.4cm}|p{0.4cm}|p{0.4cm}|p{0.4cm}}
        $i$ &1&2&3&4&5\\\hline
        probability of $\mathtt{A}$ & $\tfrac12$ & 0 & $\tfrac12$ & $\tfrac16$ & $\tfrac23$ \\\hline
        probability of $\mathtt{B}$ & $\tfrac18$ & 1 & $\tfrac38$ & $\tfrac23$ & $\tfrac13$ \\\hline
        probability of $\mathtt{C}$ & $\tfrac38$ & 0 & $\tfrac18$ & $\tfrac16$ & $0$ \\
      \end{tabular}
      \end{center}
    \caption{
      A weighted string $X'$ of length 5 over $\Sigma=\{\mathtt{A},\mathtt{B},\mathtt{C}\}$.
        \label{tab:WS2}
    }
  \end{table}

  Table~\ref{tab:big} shows how the first five maximal matching suffixes from Example~\ref{ex:max_s} are extended with position 5. Letters in gray are not part of a maximal matching suffix, but they are still present in the stream. The matching probability of the matching suffix is given for reference (in brackets); this value is not stored.
  
  Amongst the streams with stream index $(5,\mathtt{B})$, the first one has $r=0$, the second has $r=1$, and the third has $r=3$ (and, notably, it does not correspond to a maximal matching suffix). Thus, the first one of these streams is retained.
\end{example}

\begin{table}[htpb]
\begin{center}
\setlength{\tabcolsep}{0pt}
\renewcommand{\arraystretch}{1.3}
\begin{tabular}{rrp{0.2cm}rrp{0.2cm}rrp{0.2cm}rrp{0.2cm}cp{0.2cm}|p{0.2cm}rrp{0.2cm}rrp{0.2cm}rrp{0.2cm}rrp{0.2cm}rrp{0.2cm}cp{0.2cm}|p{0.2cm}l}
\multicolumn{12}{c}{before extension} & & & \multicolumn{15}{c}{after extension} & & & & & \multirow{2}{*}{stream index} \\
  \multicolumn{2}{c}{1} & & \multicolumn{2}{c}{2} & & \multicolumn{2}{c}{3} & & \multicolumn{2}{c}{4} & & & &
  & \multicolumn{2}{c}{1} & & \multicolumn{2}{c}{2} & & \multicolumn{2}{c}{3} & & \multicolumn{2}{c}{4} & & \multicolumn{2}{c}{5} & & & & & \\\hline
  \multirow{ 2}{*}{\texttt{A}} & \multirow{ 2}{*}{$\frac12$} & \multirow{ 2}{*}{\,} & \multirow{ 2}{*}{\texttt{B}} & \multirow{ 2}{*}{$1$} & \multirow{ 2}{*}{\,} & \multirow{ 2}{*}{\texttt{A}} & \multirow{ 2}{*}{$\frac12$} & \multirow{ 2}{*}{\,} & \multirow{ 2}{*}{\texttt{B}} & \multirow{ 2}{*}{$\frac23$} & \multirow{ 2}{*}{\,} & \multirow{ 2}{*}{$(\frac16)$} & \multirow{ 2}{*}{\,}
  & & \texttt{A} & $\frac12$ & & \texttt{B} & $1$ & & \texttt{A} & $\frac12$ & & \texttt{B} & $\frac23$ & & \texttt{A} & $\frac23$ & & $(\frac19)$ & & & $\emptyset$ \\
  &&&&&&&&&&&&&&
  & \textcolor{white!70!black}{\texttt{A}} & \textcolor{white!70!black}{$\frac12$} & & \texttt{B} & $1$ & & \texttt{A} & $\frac12$ & & \texttt{B} & $\frac23$ & & \underline{\texttt{B}} & $\frac13$ & & $(\frac19)$ & & & $(5,\mathtt{B})$ \\\hline
  \multirow{ 2}{*}{\texttt{A}} & \multirow{ 2}{*}{$\frac12$} & \multirow{ 2}{*}{\,} & \multirow{ 2}{*}{\texttt{B}} & \multirow{ 2}{*}{$1$} & \multirow{ 2}{*}{\,} & \multirow{ 2}{*}{\underline{\texttt{B}}} & \multirow{ 2}{*}{$\frac38$} & \multirow{ 2}{*}{\,} & \multirow{ 2}{*}{\texttt{B}} & \multirow{ 2}{*}{$\frac23$} & \multirow{ 2}{*}{\,} & \multirow{ 2}{*}{$(\frac18)$} & \multirow{ 2}{*}{\,}
  & & \texttt{A} & $\frac12$ & & \texttt{B} & $1$ & & \underline{\texttt{B}} & $\frac38$ & & \texttt{B} & $\frac23$ & & \texttt{A} & $\frac23$ & & $(\frac{1}{12})$ & & & $(3,\mathtt{B})$ \\
  &&&&&&&&&&&&&&
  & \textcolor{white!70!black}{\texttt{A}} & \textcolor{white!70!black}{$\frac12$} & & \texttt{B} & $1$ & & \underline{\texttt{B}} & $\frac38$ & & \texttt{B} & $\frac23$ & & \underline{\texttt{B}} & $\frac13$ & & $(\frac{1}{12})$ & & & $(3,\mathtt{B}),(5,\mathtt{B})$ \\\hline
  \multirow{ 2}{*}{\underline{\texttt{C}}} & \multirow{ 2}{*}{$\frac38$} & \multirow{ 2}{*}{\,} & \multirow{ 2}{*}{\texttt{B}} & \multirow{ 2}{*}{$1$} & \multirow{ 2}{*}{\,} & \multirow{ 2}{*}{\texttt{A}} & \multirow{ 2}{*}{$\frac12$} & \multirow{ 2}{*}{\,} & \multirow{ 2}{*}{\texttt{B}} & \multirow{ 2}{*}{$\frac23$} & \multirow{ 2}{*}{\,} & \multirow{ 2}{*}{$(\frac18)$} & \multirow{ 2}{*}{\,}
  & & \underline{\texttt{C}} & $\frac38$ & & \texttt{B} & $1$ & & \texttt{A} & $\frac12$ & & \texttt{B} & $\frac23$ & & \texttt{A} & $\frac23$ & & $(\frac{1}{12})$ & & & $(1,\mathtt{C})$ \\
  &&&&&&&&&&&&&&
  & \textcolor{white!70!black}{\underline{\texttt{C}}} & \textcolor{white!70!black}{$\frac38$} & & \texttt{B} & $1$ & & \texttt{A} & $\frac12$ & & \texttt{B} & $\frac23$ & & \underline{\texttt{B}} & $\frac13$ & & $(\frac{1}{9})$ & & & $(5,\mathtt{B})$, duplicate \\\hline
  \multirow{ 2}{*}{\underline{\texttt{C}}} & \multirow{ 2}{*}{$\frac38$} & \multirow{ 2}{*}{\,} & \multirow{ 2}{*}{\texttt{B}} & \multirow{ 2}{*}{$1$} & \multirow{ 2}{*}{\,} & \multirow{ 2}{*}{\underline{\texttt{B}}} & \multirow{ 2}{*}{$\frac38$} & \multirow{ 2}{*}{\,} & \multirow{ 2}{*}{\texttt{B}} & \multirow{ 2}{*}{$\frac23$} & \multirow{ 2}{*}{\,} & \multirow{ 2}{*}{$(\frac{3}{32})$} & \multirow{ 2}{*}{\,}
  & & \underline{\texttt{C}} & $\frac38$ & & \texttt{B} & $1$ & & \underline{\texttt{B}} & $\frac38$ & & \texttt{B} & $\frac23$ & & \texttt{A} & $\frac23$ & & $(\frac{1}{16})$ & & & $(1,\mathtt{C}),(3,\mathtt{B})$ \\
  &&&&&&&&&&&&&&
  & \textcolor{white!70!black}{\underline{\texttt{C}}} & \textcolor{white!70!black}{$\frac38$} & & \texttt{B} & $1$ & & \underline{\texttt{B}} & $\frac38$ & & \texttt{B} & $\frac23$ & & \underline{\texttt{B}} & $\frac13$ & & $(\frac{1}{12})$ & & & $(3,\mathtt{B}),(5,\mathtt{B})$, duplicate \\\hline
  \multirow{ 2}{*}{\textcolor{white!70!black}{\texttt{A}}} & \multirow{ 2}{*}{\textcolor{white!70!black}{$\frac12$}} & \multirow{ 2}{*}{\,} & \multirow{ 2}{*}{\texttt{B}} & \multirow{ 2}{*}{$1$} & \multirow{ 2}{*}{\,} & \multirow{ 2}{*}{\underline{\texttt{C}}} & \multirow{ 2}{*}{$\frac12$} &\multirow{ 2}{*}{\,} & \multirow{ 2}{*}{\texttt{B}} & \multirow{ 2}{*}{$\frac23$} & \multirow{ 2}{*}{\,} & \multirow{ 2}{*}{$(\frac{1}{12})$} & \multirow{ 2}{*}{\,}
  & & \textcolor{white!70!black}{\texttt{A}} & \textcolor{white!70!black}{$\frac12$} & & \texttt{B} & $1$ & & \underline{\texttt{C}} & $\frac18$ & & \texttt{B} & $\frac23$ & & \texttt{A} & $\frac23$ & & $(\frac{1}{16})$ & & & $(3,\mathtt{C})$ \\
  &&&&&&&&&&&&&&
  & \textcolor{white!70!black}{\texttt{A}} & \textcolor{white!70!black}{$\frac12$} & & \textcolor{white!70!black}{\texttt{B}} & \textcolor{white!70!black}{$1$} & & \textcolor{white!70!black}{\underline{\texttt{C}}} & \textcolor{white!70!black}{$\frac18$} & & \texttt{B} & $\frac23$ & & \underline{\texttt{B}} & $\frac13$ & & $(\frac29)$ & & & $(5,\mathtt{B})$, duplicate
\end{tabular}
\end{center}
\caption{Extensions of maximal matching suffixes from Example~\ref{ex:max_s} by position 5 with letter probabilities \texttt{A}: $\frac23$, \texttt{B}: $\frac13$, \texttt{C}: $0$.}
\label{tab:big}
\end{table}

Finally, for each of the streams we store a data structure for the approximate solution of \SWP. It takes $\Oh(\log_{1/(1-\eps)} z)$ space per stream, so storing and updating all of them upon a letter arrival takes $\Oh(z \log_{1/(1-\eps)} z)$ space and time.

\subsubsection{Approximate solution to \SWP} Recall that in the \SWP problem we are given a stream of numbers from $[0,1]$ and an integer $m$. Each time a new number arrives, we must update and output the product of numbers in the $m$-length suffix of the stream. We give a $(1-\eps)$-approximation algorithm for the problem. The algorithm may output either a number or ``No''. If it outputs a number $y$, then the product of the numbers in the $m$-length suffix of the stream is between $y$ and $y/(1-\eps)$. Otherwise, the product is less than $(1-\eps)/z$. 

\begin{lemma}\label{lm:SWP}
For a stream of numbers $\{x_i\}_{i = 1}^\infty$, where $x_i \in [0,1]$,  arriving one at a time, and a window width $m$, there is a deterministic $(1-\eps)$-approximation algorithm that takes $\Oh(\log_{1/(1-\eps)} z )$ space and $\Oh(1)$ time per arrival and solves the \SWP problem. 
\end{lemma}
\begin{proof}
At time $q$ the algorithm maintains a queue that consists of at most $M(z)=2\ceil{\log_{1-\eps} ((1-\eps)/z)}$ intervals $[i_1, i_2-1], [i_2, i_3-1], \ldots, [i_{k(q)}, i_{k(q)+1}-1]$ where $i_{k(q)+1}=q+1$, such that $1 \le i_1 < i_2 < \ldots < i_{k(q)} < i_{k(q)+1}$, which obeys the following invariant.

\begin{enumerate}[(i)]
  \item All intervals $[i_j,i_{j+1}-1]$ for $j \ge 2$ are subintervals of $[q-m+1,q]$, and $[i_1,i_2-1]$ has a non-empty intersection with $[q-m+1,q]$;
  \item If $x_{i_j} < 1-\eps$, then $i_{j+1}=i_j+1$;
  \item Otherwise, $x_{i_j}\cdot x_{i_j+1} \cdot \ldots \cdot x_{i_{j+1}-1} \ge 1-\eps$ and either $j = k(q)$ or $x_{i_j}\cdot x_{i_j+1} \cdot \ldots \cdot x_{i_{j+1}}<1-\eps$.
\end{enumerate}

The algorithm stores the product of numbers in each interval and the total product of numbers in all intervals. When $x_{q+1}$ arrives, it updates the family of intervals in the following way. Let $\pi$ be the product of numbers in the interval $[i_{k(q)}, q]$. If $\pi \cdot x_{q+1} \ge 1-\eps$, then the algorithm extends the interval $[i_{k(q)}, q]$ by the element $x_{q+1}$. Otherwise, it creates a new interval $[q+1,q+1]$. If the number of intervals becomes larger than $M(z)$ or if $i_2 \le q-m$, the algorithm deletes the leftmost interval $[i_1, i_2-1]$. Finally, it updates the total product of numbers in the intervals. Let us now explain how the algorithm exploits the intervals. 

Assume that $q \ge m$. Note that the product of numbers in each two consecutive intervals is at most $1-\eps$. Therefore, if $q-m+1 < i_1$, then the product of the last $m$ numbers is smaller than
$$(1-\eps)^{\floor{\frac12 M(z)}} = (1-\eps)^{\ceil{\log_{1-\eps} ((1-\eps)/z)}} \le (1-\eps)^{\log_{1-\eps} ((1-\eps)/z)} = \tfrac{1-\eps}{z}$$
and the algorithm outputs ``No''. Otherwise from the invariant it follows that $q-m+1 \in [i_1, i_2-1]$. In this case we output the total product of the numbers in the intervals. If $[i_1, i_2-1]$ is a singleton interval (i.e.\ $i_1 = i_2-1$), then $[i_1,i_{k(q)+1}-1] = [q-m+1,q]$, and the output equals $x_{q-m+1} \cdot x_{q-m+2} \cdot \ldots \cdot x_q$ exactly. If $[i_1, i_2-1]$ is not a singleton interval, then the product of numbers in it is at least $1-\eps$. Therefore, the output will be a $(1-\eps)$-approximation of $x_{q-m+1} \cdot x_{q-m+2} \cdot \ldots \cdot x_q$. The space complexity follows from the fact that $M(z) = \Oh(\log_{1/(1-\eps)} z)$. 
\end{proof}

\begin{lemma}\label{lem:SWP_str}
  Consider the stream $\{x_i\}_{i = 1}^q$ related to maximal matching suffix $S$ of $T[1,q]$ and let $q \ge m$.
  If the algorithm for \SWP returns ``No'' or a number that is smaller than $(1-\eps)/z$,
  then either $|S|<m$ or the probability $\pi$ that $S[|S|-m+1,|S|]$ matches $T[q-m+1,q]$ is below $1/z$.
  Otherwise, $|S| \ge m$ and the result $y$ of the algorithm satisfies $y \le \pi \le y/(1-\eps)$.
\end{lemma}
\begin{proof}
  Let us denote $\pi'=x_{q-m+1} \cdot \ldots \cdot x_q$.
  If the \SWP algorithm returns ``No'', then $\pi' < (1-\eps)/z$.
  Then, indeed, either $|S| < m$ or $\pi=\pi'<(1-\eps)/z$.

  Assume that the \SWP algorithm returns a number~$y$.
  Then $y \le \pi' \le y/(1-\eps)$.
  If $y<(1-\eps)/z$, then $\pi'<1/z$.
  Again this means that either $|S| < m$ or $\pi=\pi'<1/z$.
  
  Finally, consider the case that $y \ge (1-\eps)/z$.
  Since $\pi' \ge (1-\eps)/z \ge 1/(2z)$, we have $|S| \ge m$.
  This concludes that $\pi=\pi'$ and $y \le \pi \le y/(1-\eps)$, as required.
\end{proof}

\subsubsection{Finding a maximal matching suffix that ends with $P$} We give two different methods for this task. The first method is based on a $k$-\Mis algorithm. Recall that if $P$ matches $T$ at some alignment $q$ with probability at least $1/z$, then $\heavy T[q-m+1,q]$ is a $\log z$-mismatch occurrence of~$P$ (see Observation~\ref{obs:logz}). We use the $k$-\Mis algorithm with $k = \log z$ for $P$ and $\heavy{T}$ to find all such alignments. When we identify a $\log z$-mismatch occurrence of $P$ in $\heavy{T}$, we find a stream that corresponds to a maximal matching suffix that ends with $P$, if any (using the indexing trie). Suppose we have found a maximal matching suffix $S$ that ends with $P$. We then use the algorithm of Lemma~\ref{lm:SWP} to compute a $(1-\eps)$-approximation of the product of the last $m$ probabilities in the stream associated with $S$. By Lemma~\ref{lem:SWP_str}, we can output it as an answer. 

By Observation~\ref{obs:Amir}, the number of maximal matching suffixes of $T[1,q]$ never exceeds $2z$. Therefore, the total number of streams is $\Oh(z)$. Updating the streams, including removing the duplicates, takes $\Oh(z \log z)$ space and time per position. For each of the streams we run the \SWP algorithm, which takes $\Oh(z \log_{1/(1-\eps)} z)$ space and time per position (Lemma~\ref{lm:SWP}). For $k = \log z$, the $k$-\Mis algorithm takes $\S_{\log z}$ space and $\T_{\log z}$ time per arriving symbol.  Finally, to find the right stream we need just $\Oh(\log z)$ additional time per position. In total, this is $\Oh(z\log_{1/(1-\eps)} z) + \S_{\log z}$ space and $\Oh(z \log_{1/(1-\eps)} z) + \T_{\log z}$ time per position. Our algorithm can output an incorrect answer only when the $k$-\Mis algorithm errs, which happens with probability $\PP_{\log z}$. We arrive at the second claim of Theorem~\ref{th:WPM-kmismatch}:

\begin{proposition}
If only the text is weighted, there is a $(1-\eps)$-approximation streaming algorithm that solves \WPM in $\Oh(z \log_{1/(1-\eps)} z)+ \S_{\log z}$ space and $\Oh(z \log_{1/(1-\eps)} z) + \T_{\log z}$ time per arrival. At each arrival, the algorithm can err with probability~$z \cdot \PP_{\log z}$. 
\end{proposition}

The second method is based on a straightforward application of the streaming \EPM algorithm of Breslauer and Galil~\cite{DBLP:journals/talg/BreslauerG14} for the pattern $P$ and each of the maximal matching suffixes streams. The streaming \EPM algorithm of Breslauer and Galil~\cite{DBLP:journals/talg/BreslauerG14} uses $\Oh(\log m)$ space and takes $\Oh(1)$ time per arrival. Once we have found a maximal matching suffix that ends with $P$, we apply the algorithm of Lemma~\ref{lm:SWP} to compute a $(1-\eps)$-approximation of the product of the last $m$ probabilities in the stream associated with $S$ and output it as an answer. In total, the second approach requires $\Oh(z \log_{1/(1-\eps)} z + z \log m)$ space and $\Oh(z \log_{1/(1-\eps)} z)$ time per position. The error probability of the \EPM algorithm and, therefore, of our approach is $1/\poly(n)$ per position. We arrive at the second claim of Theorem~\ref{th:WPM-dictionary}:

\begin{proposition}
If only the text is weighted, there is a $(1-\eps)$-approximation streaming algorithm that solves \WPM in $\Oh(z (\log_{1/(1-\eps)} z +\log m))$ space and $\Oh(z \log_{1/(1-\eps)} z)$ time per arrival. At each arrival, the algorithm can err with probability~$1/\poly(n)$. 
\end{proposition}

\subsection{Case 3~--- both the text and the pattern are weighted}	
\label{sec:WPWT}
We conclude by showing two solutions for the case when both the pattern and the text are weighted. Our solutions combine the ideas of the two previous sections. Recall that at each alignment the algorithm must output ``Yes'' if there is a regular string that matches the text and the pattern with probability above $1/z$. The algorithm may also output ``Yes" if there is a string that matches the pattern with probability at least $1/z$ and the text with probability between $(1-\eps)/z$ and $1/z$. Otherwise it outputs ``No''. 

In the algorithm we maintain a set of at most $2z$ streams for the maximal matching suffixes of the text $T$. For each of the streams we run the \SWP algorithm (Lemma~\ref{lm:SWP}). Let us recall that the set $D$ contains at most $z$ regular strings that match $P$ with probability $\ge 1/z$. There is a regular string that matches the pattern and the text with probability $\ge 1/z$ if and only if the following two conditions hold: 

\begin{enumerate}[(i)]
\item One of the maximal matching suffixes ends with a string in $D$;
\item The product of the last $m$ probabilities in the suffix's stream is $\ge 1/z$.
\end{enumerate}

We check the condition (ii) approximately using the \SWP algorithm. For the condition (i), we again discuss two methods. 

The first method uses a $k$-\Mis algorithm. If $P$ matches $T$ at some alignment $q$, then $\heavy T[q-m+1,q]$ must be a $\log z$-occurrence of one of the strings we generated for the pattern. We use the $k$-\Mis algorithm with $k = \log z$ for each of the generated strings and $\heavy{T}$ to find all such occurrences. At each alignment $q$ we consider $\log z$-occurrences of the generated strings and build a trie of their indices, where the indices are defined as in Section~\ref{sec:SPWT}. If a maximal matching suffix ends with one of these strings, its index must be in the trie. We can perform a search for each index in the trie in $\Oh(z \log z)$ time in total. The time required to update the suffix streams is $\Oh(z \log z)$ as there are $\Oh(z)$ of them~\cite{amir_weighted_property_matching_j}. For each of the streams we run the \SWP algorithm, which takes $\Oh(z \log_{1/(1-\eps)} z)$ space and time. For $k = \log z$ the $k$-\Mis algorithm takes $\S_{\log z}$ space and $\T_{\log z}$ time per symbol. Building the trie of indices takes $\Oh(z \log z)$ time. Searching for the indices in the trie takes $\Oh(z \log z)$ time. In total, this is $\Oh(z \log_{1/(1-\eps)} z) + z \cdot \S_{\log z}$ space and $\Oh(z \log_{1/(1-\eps)} z) + z \cdot \T_{\log z}$ time. The algorithm can output an incorrect answer only when one of the $k$-\Mis algorithms errs, which happens with probability $z/\poly(m) = 1/\poly(m)$. We arrive at the following proposition. It is not included in Theorem~\ref{th:WPM-kmismatch} since the complexities resulting for the currently best values of $\S_{\log z}$ and $\T_{\log z}$ are worse than if the second method is used.

\begin{proposition}
When both the pattern and the text are weighted, there is a $(1-\eps)$-approximation streaming algorithm that solves the \WPM problem in $\Oh(z\log_{1/(1-\eps)} z) + z \S_{\log z}$ space and $\Oh(z\log_{1/(1-\eps)} z) + z \T_{\log z}$ time per arrival. At each arrival, the algorithm can err with probability~$z \cdot \PP_{\log z}$. 
\end{proposition}

The second method is to run the \DM algorithm for multiple streams~\cite{DBLP:conf/icalp/GolanKP18}. This algorithm takes a dictionary and a set of streaming texts as an input. When a new symbol of a text arrives, the algorithm must determine whether the current text ends with a string from the dictionary. We run this algorithm for the dictionary $D$ and each of the maximal matching suffixes streams. The algorithm uses $\Oh(z \log m)$ shared memory, $\Oh(\log m \log z)$ space per stream, and $\Oh(\log m)$ time per each arriving character of a stream. In our case the total space is thus $\Oh(z \log m \log z)$ and, since the maximal matching suffix streams are handled in a dynamic way, the algorithm takes $\Oh(z \log m \log z)$ time per each arriving character of the text. The error probability is $1/\poly(n)$. We therefore arrive at the third claim of Theorem~\ref{th:WPM-dictionary}:

\begin{proposition}
When both the pattern and the text are weighted, there is a $(1-\eps)$-approximation streaming algorithm that solves the \WPM problem in $\Oh(z(\log_{1/(1-\eps)} z + \log z \log m))$ space and $\Oh(z (\log_{1/(1-\eps)} z + \log z\log m))$ time per arrival. At each arrival, the algorithm can err with probability~$1/\poly(n)$. 
\end{proposition}

\subsection{Space lower bound~--- Proof of Proposition~\ref{prp:WPT-lb}}\label{subsec:prop}
\label{sec:WPT-lb}
The Yao minimax principle~\cite{minimax} implies that the expected space complexity of the optimal deterministic
algorithm for an arbitrarily chosen input distribution $\pi$ is a lower bound on the space complexity of the optimal Monte Carlo randomized algorithm. Therefore, it suffices to show that for some input distribution $\pi$ any deterministic $(1-\eps)$-approximate streaming algorithm for the \WPM problem requires $\Omega(z)$ space on average (the lower bound for the exact algorithm follows). 

\paragraph{Case 1: Only pattern is weighted}
Let us consider the following communication problem. Alice is given a weighted pattern $P$ of length $m$ drawn from some distribution and the prefix $T[1,m]$ of the regular text $T$ of length $2m-1$. Bob is given the suffix $T[m+1,2m-1]$. Alice and Bob are given a threshold $1/z$ and Bob must distinguish all alignments of $P$ and $T$ where the match probability is at most $(1-\eps)/z$ and at least $1/z$. The communication complexity of the problem is defined to be the expectation of the minimum number of bits that Alice must send to Bob so that he can achieve his mission. 

The lower bound on the communication complexity of this problem is a lower bound on the space complexity of any streaming algorithm for the \WPM by a standard reduction. Namely, Alice can run the algorithm on the pattern and her half of the text and then send the memory of the algorithm to Bob. The space complexity of the algorithm will be equal to the communication complexity of the problem.

We select a string $S$ of length $m$ over an alphabet $\{0,1\}$ uniformly at random and build a weighted pattern over an alphabet $\{\mathtt{A},\mathtt{B},\mathtt{C}\}$. For $1 \le i \le m$, if $S[i] = 0$, then $P[i] = \mathtt{A}$ with probability $1-\frac1m$, $\mathtt{B}$ with probability $\frac2{3m}$, and $\mathtt{C}$ with probability $\frac1{3m}$. If $S[i] = 1$, then $P[i] = \mathtt{A}$ with probability $1-\frac1m$, $\mathtt{B}$ with probability $\frac1{3m}$, and $\mathtt{C}$ with probability $\frac2{3m}$. The text is defined to be $\mathtt{A}^{m-1} \mathtt{B} \mathtt{A}^{m-1}$ and $1/z = \frac2{3m} (1-\frac1m)^{m-1} = \Theta(\frac1m)$. We choose $\eps$ to be an arbitrary constant in $[0,\frac12]$. If Bob can distinguish all alignments of $P$ and $T$ where the match probability is at most $(1-\eps)/z$ or at least $1/z$, he can also restore $S$. Indeed, if the match probability between $P$ and $T[i,i+m-1]$ is equal to $\frac2{3m}(1-\frac1m)^{m-1} = 1/z$, then $S[m+1-i] = 0$, and if the match probability between $P$ and $T[i,i+m-1]$ is equal to $\frac1{3m}(1-\frac1m)^{m-1} \le (1-\eps)/z$, then $S[m+1-i] = 1$. Therefore, from information theoretic ideas, the communication complexity of the problem must be $\Omega(m) = \Omega(z)$.  

\paragraph{Case 2: Only text is weighted}
For the case when only the text is weighted, we use the same reduction but define $P$ and $T$ differently. Let $S$ be a random string in $\{0,1\}^m$ as above. $P$ is defined to be $\mathtt{B} \mathtt{A}^{m-1}$ and $T$ is defined over $\{\mathtt{A},\mathtt{B},\mathtt{C}\}$ as follows. For $1 \le i \le m$, if $S[i] = 0$, then $T[i] = \mathtt{A}$ with probability $1-\frac1m$, $T[i] = \mathtt{B}$ with probability $\frac2{3m}$, and $T[i] = \mathtt{C}$ with probability $\frac1{3m}$. If $S[i] = 1$, then $T[i] = \mathtt{A}$ with probability $1-\frac1m$, $T[i] = \mathtt{B}$ with probability $\frac1{3m}$, and $T[i] = \mathtt{C}$ with probability $\frac2{3m}$. For $m+1 \le i \le 2m-1$, $T[i] = \mathtt{A}$ with probability $1-\frac1m$ and $T[i] = \mathtt{B}$ with probability $\frac1m$. Again, if Bob can distinguish all alignments of $P$ and $T$ where the match probability is at most $(1-\eps)/z$ or at least $1/z$, he can also restore $S$. Therefore, any streaming algorithm for this variant must use $\Omega(z)$ space as well.

\paragraph{Case 3: Both the text and the pattern are weighted}
As a corollary we immediately obtain that any exact or $(1-\eps)$-approximate streaming algorithm for the weighted pattern, weighted text case must use $\Omega(z)$ space. 

\begin{remark}
We gave the proof for $z = \Theta(m)$. However, we can modify the examples so that it holds for $z \le m$: it suffices to prepend the pattern and the text with $\mathtt{A}^k$, for sufficiently large $k$, which for a weighted sequence means adding $k$ positions with probability of $\mathtt{A}$ equal to 1.
\end{remark}

\section{Conclusions}
In this work we present the first efficient solutions to the problems of $k$-\MwDR and \WPM in the streaming model.
In parallel to our work, a yet more efficient algorithm for $k$-\MwDR was developed~\cite{DBLP:conf/soda/CliffordKP19}
that near-matches time and space lower bounds.
We provide lower bounds for streaming solutions to \WPM that show that the space complexity of our solutions is nearly optimal.
It is an interesting open problem if the two variants of \WPM where the text is weighted can be solved exactly using
$\Oh(z^{1-\eps})\cdot \log^{\Oh(1)}m$ time per arrival, for any $\eps>0$.

\paragraph{Acknowledgements}
This work was supported by the ``Algorithms for text processing with errors and uncertainties''
project carried out within the HOMING programme of the Foundation for Polish
Science co-financed by the European Union under the European Regional Development Fund.

\bibliographystyle{plainurl}
\bibliography{main}
\end{document}